%% file: coun-inv.tex
\documentclass{sig-alternate-10}

\usepackage{euscript}
\usepackage{graphicx}
\usepackage{hyphenat}
\usepackage{microtype} 
\usepackage{multirow}
\usepackage{enumitem}

\usepackage{xspace}
\usepackage{verbatim}
\usepackage{color}

\input{yf-def}

\def\S{\mathcal{S}}

\begin{document}

\conferenceinfo{}{}

\title{New Results for Adaptive and Approximate Counting of Inversions}
\numberofauthors{1}

\author
{
\alignauthor
Saladi Rahul \\[2mm]
\affaddr{Dept. of Computer Science and Engg., University of Minnesota Twin-Cities,\\
  4-192 Keller Hall, 200 Union St. S.E., Minneapolis, MN 55455, USA}
 \email{sala0198@umn.edu} 
}

\maketitle

\begin{abstract}
Counting inversions is a classic and important problem in databases.
The number of inversions, $K^*$,  in a list $L=(L(1),L(2),\ldots,L(n))$ is defined as the number of pairs $i < j$ with $L(i) > L(j)$. 
In this paper, new results for this problem are presented:

\begin{enumerate}
 \item In the I/O-model, an {\it adaptive} algorithm is presented for calculating $K^{*}$. 
 The algorithm performs $O(\frac{N}{B}+ \frac{N}{B}\log_{M/B}(\frac{K^*}{NB}))$ I/Os. When $K^{*}=O(NM)$, then 
 the algorithm takes only $O(\frac{N}{B})$ I/Os. This algorithm can be modified to match the 
 state of the art for the comparison based model and the RAM model.

 \item In the RAM model, a linear-time algorithm is presented to obtain a tight estimate of $K^*$; 
 specifically a value  which lies with 
 high probability in the range  \\ $[(1-\frac{\log N}{N^{1/4}})K^*,(1+\frac{\log N}{N^{1/4}})K^*]$. The state of the art linear-time 
 algorithm works for the {\em special} case where $L$ is a permutation, i.e., 
 each $L(i)$ is a distinct integer in the range $[1,N]$.
 In this paper, we handle a {\em general} case
 where each $L(i)$ is a real number. 
  
 \end{enumerate}
\end{abstract}

\section{Introduction}

In this paper we revisit the classic database problem of counting inversions. The number of 
inversions, $K^*$, in a list $L=(L(1), L(2),\ldots, L(N))$ is defined as the number 
of pairs $i < j$ with $L(i) > L(j)$. Each value $L(i)$ is a real number.

\subsection{Motivation}

\noindent
{\em Classical motivation.} Interest in studying the counting inversions problem has been shown by various communities in computer science. 
It is considered an important measure to test the ``sortedness'' of the data. 
For example, sorting data is a critical operation in large-scale applications. Typically, such applications 
have multiple sorting algorithms and they perform some ``tests'' on the data to decide the most suitable 
sorting algorithm (an insertion-sort type algorithm is fast if the data is almost sorted). One of the important test happens to be counting inversions . We refer the reader 
to the book of Knuth \cite{knuth3} and the 
survey report of Estivill-Castro and Wood \cite{ew92} for a detailed discussion on how counting inversions is crucial to 
the engineering of a fast sorting algorithm.

\vspace{0.1 in}
\noindent
{\em Modern motivation.} Modern applications have revised the interest in the problem of counting inversions.
We  briefly mention the applications here:  (a) The number of inversions between two permutations is important for {\em rank aggregation} in Internet-based applications \cite{dkns01}, 
and (b) The robustness of a {\em ranking function} (of database entries) can be tested via counting inversions.
We strongly refer the reader to Ajtai {\em et al.} \cite{ajks02} for a nice detailed description of how modern applications benefit from counting inversions.

\subsection{Previous work on counting inversions}

{\em Non-adaptive algorithms.} The standard textbook solution for the counting inversions problem takes $O(n\log n)$ time by mergesort. 
. There have been improvements over the $O(n\log n)$ time algorithm in the RAM model.  
Using Dietz's dynamic ranking structure \cite{d89}  counting inversions can be done in $O(n\log n/\log \log n)$ time. Few years back, Chan and 
Patrsacu \cite{cp10} could significantly improve the running time to $O(n\sqrt{\log n})$. Interestingly, in the RAM model counting inversions 
seems to be harder than sorting: the best known deterministic sorting algorithm takes $O(N\log\log N)$ time~\cite{h04} 
and the best known randomized sorting algorithm takes $O(N\sqrt{\log\log N})$ expected time~\cite{ht02}.

{\em Adaptive algorithms.} One  approach to develop faster algorithms is to build solutions which adapt based on the number 
of inversions. Mehlhorn~\cite{m79} presented an $O\left(N + N\log\left(\frac{K^{*}}{N}\right)\right)$ time algorithm to count inversions 
in the comparison based model. Adapting the approach of Pagh, Pagh, and Thorup~\cite{ppt04}, Elmasry~\cite{e15} presented an 
$O\left(N + N\sqrt{\log\left(\frac{K^{*}}{N}\right)}\right)$ time algorithm in the RAM model.

{\em Approximate algorithms.} The other approach to develop faster algorithms to count inversions is to approximate the value of $K^{*}$. 
To obtain faster algorithms, Andersson and Petersson \cite{ap98}, and Chan and Patrascu \cite{cp10} studied the approximate version of counting inversions problem. 
If the number of inversion in the list is $K^{*}$, then their algorithm will report a value within an additive error of $\varepsilon K^{*}$. 

{\em Streaming setting.} The focus of this paper is the RAM model and the I/O-model. 
However, there are other interesting models in which this problem has been studied.
For example, the last decade saw the streaming community getting interested \cite{ajks02,gz03}.

\subsection{Our Results} 
In this paper, we present two new results on the problem of counting inversions.

\vspace{0.1 in}
\noindent
{\em Adaptive algorithm.} In the I/O-model we present an adaptive algorithm which counts the 
number of inversions  using $O(\frac{N}{B} + \frac{N}{B}\log_{M/B}(\frac{K^{*}}{NB}))$ I/Os.
Previously, such adaptive algorithms were known only in the comparison model~\cite{m79} and the RAM model~\cite{e15}. 
Neither of these solution can be trivially modified to work efficiently in the I/O-model. 
For example, adapting the algorithm of \cite{e15} to the I/O-model requires $\Omega(N)$ I/Os, since it inserts one 
point at a time. 
Interestingly, our algorithm can be modified to match the 
 state of the art for the comparison based model and the RAM model.
 In that sense, our algorithm subsumes the results of \cite{e15,m79}.
 Please see the appendix for a brief description of the I/O-model.

\vspace{0.1 in}
\noindent
{\em Approximate algorithm.} This problem is studied in the RAM model. 
We present an $O(N)$ time algorithm which reports a value in the range \\
$\left[ \left(1-\frac{\log N}{N^{1/4}}\right)K^{*},\left(1+ \frac{\log N}{N^{1/4}}\right)K^{*}\right]$. 
The estimate is correct with probability $1-1/N^c$, where $c$ is a constant independent of $N$.

Chan and Patrascu \cite{cp10} also presented an $O(N)$ time algorithm for this problem. 
However, their solution works only for the special case where $L$ is a permutation, i.e., 
each $L(i)$ is a distinct integer in the range $[1,N]$. Because they consider a permutation, 
they  make use of the Spearman's Footrule~\cite{dg77} which already gives a $2$-factor approximation of 
$K^{*}$. In this paper, we study the more challenging setting where each element in $L(i)$ is a 
real number. A new approach is needed to handle this setting.

\section{Red-blue dominance counting}

We start by defining the {\em red/blue dominance counting problem}. 
We are given a red list $R=(R(1), R(2),\ldots, R(n))$ and a blue list $B=(B(1), B(2),\ldots, B(n))$.
Each element in $R$ is mapped to a two-dimensional point: $R(i)$ is mapped to a point 
$(i,R(i))$. Similarly, each element, say $B(i)$, in $B$ mapped to a point $(i, B(i))$. 
A pair $(r,b)$ is called an {\em domination pair} if $r$ is a red point dominated by a blue point $b$. 
A blue point $b$ dominates a red point $r$ if $b$ has a larger  $x$-coordinate than $r$ and 
$b$ has a  smaller $y$-coordinate than $r$ (see Figure~\ref{fig:main-figure}(a)). 

Throughout the paper, we will interpret $R$ and $B$ as one of the following: 
(1) a list of $N$ elements storing real-values, or (2) a pointset in two-dimensional plane. 
It will be clear from the context which interpretation is being taken.

Let $K^{*}$ be the number of domination pairs in $R$ and $B$. 
Counting inversions is a special case of this problem by letting the red point set be equal to the blue point set.
In this paper, we present two results for the  red-blue dominance counting problem.
\begin{theorem}\label{thm:main-1}
{\em (Adaptive algorithm)} Red-blue dominance counting problem can be solved using \\ $O(\frac{N}{B} + \frac{N}{B}\log_{M/B}(\frac{K^{*}}{NB}))$ I/Os, 
where $K^{*}$ is the number of domination pairs. 
When $K^{*}=O(NM)$, then the algorithm uses only $O(\frac{N}{B})$ I/Os.
This problem is studied in the $I/O$-model.
\end{theorem}

\begin{theorem}\label{thm:approx}
{\em (Approximate algorithm)} Red-blue approximate dominance counting problem can be solved in $O(N)$ time.
For a fixed constant $c$, with probability $1-1/N^c$ the algorithm will report a value in the range \\ 
$\left[ \left(1-\frac{\log N}{N^{1/4}}\right)K^{*},\left(1+ \frac{\log N}{N^{1/4}}\right)K^{*}\right]$. 
This problem is studied in the RAM-model.
\end{theorem}

\begin{figure*}[t]
 \centering
\includegraphics[scale=1]{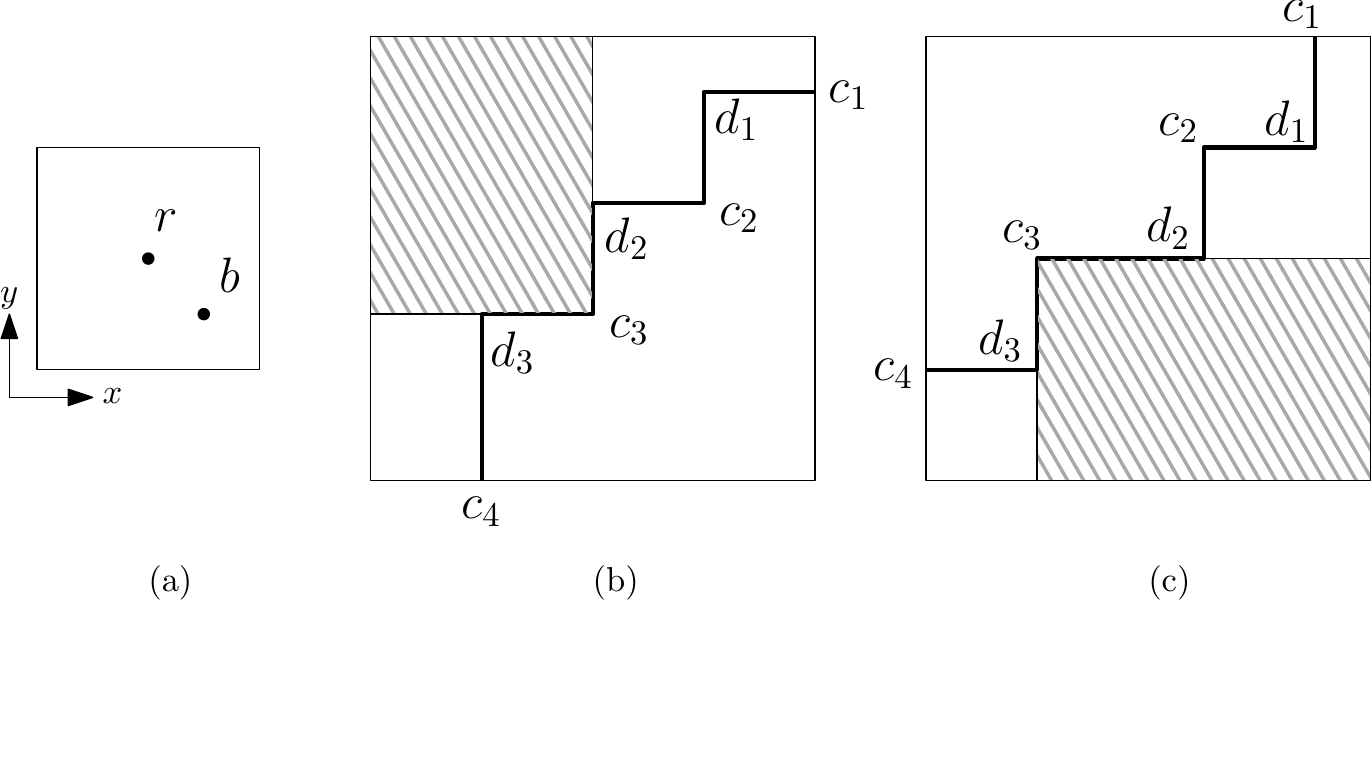}
\vspace{-0.5 in}
\caption{}
\label{fig:main-figure}
\end{figure*}

In Section~\ref{sec:red-blue} we will define the concept of red-blue cells along with their 
properties. At first look, it might not be clear to the reader as to why we need red-blue cells. 
Then in Section~\ref{sec:adaptive} we will make use of them to obtain the adaptive algorithm 
and then in Section~\ref{sec:approx} we will use them along with random sampling techniques to 
obtain the approximate algorithm.
\section{Construction of Red-Blue Cells}\label{sec:red-blue}
Given the lists $R$ and  $B$, and a parameter $K$, we want to construct 
a set of {\em red-blue cells} $C_1, C_2,\ldots,C_\ell$. 
 A red-cell is a rectangle of the form 
$(-\infty,x) \times (y,\infty)$, and a blue-cell of the form $(x,\infty) \times (-\infty,y)$.
With each cell $C_i$ we associate a set of red points $R_i \subseteq R$ and 
a set of blue points $B_i \subseteq B$. 
Consider the following two cases:

\vspace{0.1 in}
\noindent
(1) If $K^{*} \leq K$, then we want to construct 
red-blue cells which satisfy the following {\em three} properties:
\begin{enumerate}[label=\Alph*)]
\item $\forall i \in [1,\ell], \min\{|R_i|,|B_i|\}=O(K/N)$.
\item For every domination pair $(r,b)$ there will exist exactly a single integer $i$  such that 
$b\in B_i$ and $r \in R_i$.
\item $\sum_{i=1}^{\ell} |R_i|=O(N)$, and $\sum_{i=1}^{\ell} |B_i|=O(N)$.
\end{enumerate}
\vspace{0.1 in}
\noindent
(2) If $K^{*} >K$, then we either construct the cells with the properties described 
above, or we are allowed to report a {\em failure}.

\begin{lemma}\label{lem:cells}
The red-blue cells  can be constructed in $O(N/B)$ I/Os.
\end{lemma}
The rest of this section is dedicated to proving Lemma~\ref{lem:cells}.
\subsection{First step: Red cells}
{\em Shallow cuttings} for various geometric objects are widely used in computational geometry to answer range searching and related problems 
(for example, \cite{ac09,m92}). Shallow cuttings as described in this section have been used before by Vengroff and Vitter \cite{vv96}.
On the technical side, our key contribution is a novel and a non-trivial application of shallow cuttings.

Consider a red pointset  $R$. Informally, a $k$-shallow cutting on the  pointset $R$ has the form of a ``staircase'' which is a one-dimensional, monotone 
sequence of orthogonal line-segments. Formally, a $k$-shallow cutting is a curve ${\cal C}=c_1d_1c_2d_2\ldots d_{t-1}c_t$ of alternating 
horizontal line segments $c_id_i=[x(c_i),x(d_i)] \times [y(c_i)]$ and vertical line segments 
$d_ic_{i+1}=[x(d_i)] \times [y(c_{i+1}), y(d_i)]$. 
See Figure~\ref{fig:main-figure}(b). The points $c_1,c_2,\ldots,c_t$ are called {\it outward corners} and the points $d_1,d_2,\ldots,d_{t-1}$ are 
called {\it inward corners}. 
With each outward corner $c_i=(x,y)$, we associate a {\em cell} $C_i =(-\infty,x) \times (y,+\infty)$.
If a point $q$ is dominated by at least one outward corner, then $q$ is said to lie {\it above} the curve ${\cal C}$. 
On the other hand, if a point $q$ dominates at least one inward corner, then $q$ is said to lie {\it below} the curve ${\cal C}$. 
The curve ${\cal C}$ has the following properties:

\begin{enumerate}
 
 \item Every point on curve ${\cal C}$ dominates at least $k$ points in $R$, but it dominates no more than $2k$ points in $R$.
 
 \item If a point $q$ dominates less than  $k$ points of $R$, 
 then $q$ lies above the curve ${\cal C}$.
 
 \item  $t=O(n/k)$, i.e., the number of cells are no more than 
 $O(n/k)$.
\end{enumerate}

\begin{lemma}\label{lem:sc}
The $k$-shallow cutting on $R$ can be constructed using $O(N/B)$ I/Os.
The inward and the outward corners are reported in increasing order of 
their $x$-coordinate value.
\end{lemma}
\begin{proof}
There exists a simple algorithm to construct the $k$-shallow cutting on $R$. 
The details of this construction can be found in \cite{vv96}.
\end{proof}

\noindent
{\em Algorithm.} Given $R$ and $B$, we construct the first set of cells, which we 
call {\em red cells}.
\begin{enumerate}
\item Construct a $\left\lceil\frac{2K}{N}\right\rceil$-shallow cutting ${\cal C}_r$ on $R$.
\item For each blue point $b\in B$ check if it lies on/below or above the curve ${\cal C}_r$. If $b$ lies on/below ${\cal C}_r$, 
then it is classified as {\em deep}. Otherwise, it is classified as {\em shallow} and $b$ is {\em assigned} to any 
arbitrary cell in the cutting containing it.
\item  Delete all the shallow blue points from the dataset. 
If the number of deep blue points are greater than $N/2$, then we report a {\em failure} and 
stop the algorithm concluding that $K^{*} > K$. Otherwise, we go to the next step.
\item For each cell $C_i$ we define  $R_i$ to be the set of red points which lie in that cell, 
and  $B_i$ to be the set of blue points assigned to that cell. 
\end{enumerate}

\noindent
{\em Analysis.}
Now we analyze the running time of the above algorithm. Using Lemma~\ref{lem:sc}, step~$1$  can be 
performed in $O(N/B)$ I/Os. Step~$2$ 
is performed as follows: for each blue point (say $b$) find the outward corner (say $c_i$) immediately to 
its right. $b$ is assigned to cell $C_i$ if it lies in the cell of $C_i$; otherwise $b$ lies on/below ${\cal C}_r$ and is classified as deep. 
The blue points can be assigned using $O(N/B)$ I/Os since the blue points and the outward corners are already sorted along $x$-axis.

Next we show that when $K^{*} \leq K$ then the algorithm does not report a failure. 
Each deep blue point dominates $\geq \left\lceil\frac{2K}{N}\right\rceil$ red points 
(by Property~$2$ of shallow cuttings). Since there are at most $K^{*}$ domination 
pairs, the number of deep blue points is $\leq \frac{K^{*}}{\lceil\frac{2K}{N}\rceil} \leq \frac{K}{\lceil\frac{2K}{N}\rceil} \leq N/2$.  Hence, the algorithm will not report failure when $K^{*} \leq K$.

Now we prove that none of the three properties of the red-blue cells have been violated.
For each cell $C_i$, the outward corner $c_i$ dominates  $O(K/N)$ red points. Therefore, 
$|R_i|=O(K/N)$ and hence, property (A) is not violated.  By step~$4$ of our algorithm, 
we ensure property (B) for every domination pair $(r,b)$ where $b$ is a shallow point.
The other domination pairs will be taken care of in the next steps. 

Since each point in $B$ is assigned to exactly one cell, $\sum_{i=1}^{\ell} |B_i|\leq N$.
By Property~$3$ of shallow cuttings, the number of cells constructed is $O(N^2/K)$, 
and by Property~$1$ of shallow cuttings, each cell contains $O(K/N)$ red points. 
Therefore, $\sum_{i=1}^{\ell} |R_i|=O(N^2/K) \times O(K/N)=O(N)$. Therefore, property (C) 
has not yet been violated.

\vspace{0.1 in}
\noindent
{\em Remark.} Note that a red point in $R$ can belong to many $R_i$'s, whereas a shallow blue point in $B$ 
will belong to exactly one $B_i$.

\subsection{Second step: Blue cells}
After the first step, all the domination pairs $(r,b)$ involving the shallow blue points have been taken care of. 
In the next two steps, we discuss how to build additional cells which will capture domination pairs involving the deep blue points.

We will use shallow cuttings again, but this time we will change the orientation of our cells.
A $k$-shallow cutting on the deep blue points is a curve ${\cal C}=c_1d_1c_2d_2\ldots d_{t-1}c_t$ of alternating 
vertical line segments $c_id_i=[x(c_i)] \times [y(d_i),y(c_i)]$ and horizontal line segments  \\ $d_ic_{i+1}=[x(c_{i+1}),x(d_i)] \times [y(d_i)]$. 
See Figure~\ref{fig:main-figure}(c). 
The points $c_1,c_2,\ldots,c_t$ are called {\it outward corners} and the points $d_1,d_2,\ldots,d_{t-1}$ are 
called {\it inward corners}. 
With each outward corner $c_i=(x,y)$, we associate a {\em blue cell} $C_i =(x,\infty) \times (-\infty,y)$.
If a point $q$ is dominated by at least one inward corner, then $q$ is said to lie {\it above} the curve ${\cal C}$. 
On the other hand, if a point $q$ dominates at least one outward corner, then $q$ is said to lie {\it below} the curve ${\cal C}$. 
The curve ${\cal C}$ should have the following properties:

\begin{enumerate}
 
 \item Every point on curve ${\cal C}$ is dominated by at least $k$ deep blue points, but it is dominated by no more than $2k$ deep blue points.
 
 \item  If a point $q$ is dominated by less than  $k$ deep blue points, 
 then $q$ lies below the curve ${\cal C}$.
 
 \item $t=O(n/k)$.
\end{enumerate}

\noindent
{\em Algorithm.} Let $B_D$ be the set of deep blue points. Given $R$ and $B_D$, the following steps are performed:
\begin{enumerate}
\item Construct $\lceil\frac{2K}{N}\rceil$-shallow cutting ${\cal C}_b$ on all the  deep blue points.
\item For each red point $r$ check if it lies on/above or below the curve ${\cal C}_b$. If $r$ lies on/above ${\cal C}_b$, 
then it is classified as {\em deep}. Otherwise, it is classified as {\em shallow} and $r$ is assigned to any 
arbitrary cell in the cutting containing it.
\item  If the number of deep red points are greater than $N/2$, then we report a {\em failure} and 
stop the algorithm concluding that $K^{*} > K$. Otherwise, we go to the next step.
\item For each cell $C_i$ we define $B_i$ to be the set of deep blue points which lie in that cell, 
and $R_i$ to be the set of red points assigned to that cell. 
\end{enumerate}

Following the analysis from the previous step, the number of I/Os performed 
in this step is also bounded by $O(N/B)$, and it can be shown that none of 
the properties of the red-blue cells are violated yet.
\subsection{Third step: Recursion}

In the second step, all the domination pairs $(r,b)$ such that $r$ is a shallow red point and  $b$ is a 
deep blue point will be taken care of. After the first two steps, now we are left with deep red points $R_D$ 
and deep blue points $B_D$.  We know that $|R_D| < N/2$ and $|B_D| < N/2$; else a failure would have 
been reported.

\vspace{0.1 in}
\noindent
{\em Algorithm.} Recurse on $R_D$ and $B_D$, and all occurrences of $N$ in the algorithm are replaced 
with $N/2$. The algorithm stops when the red and the blue set is smaller than a suitable constant $C$.

\vspace{0.1 in}
Let $T(N)$ denote the total number of I/Os performed by this algorithm. Then,

\[ T(N) \leq \left\{\begin{array}{ll}
		  O(1) & \mbox{if $n \leq C$};\\
                  O\left(\frac{N}{B}\right) + T(N/2) & \mbox{otherwise}.
                 \end{array}\right.\]
                 
Solving this recurrence we get $T(N) = O\left(\frac{N}{B}\right)$.   
By a similar recurrence, Property~(C) of red-blue cells is satisfied. 
It is easy to verify that Property~(A) and (B) are also satisfied.             
This finishes the proof of Lemma~\ref{lem:cells}.

\section{The Adaptive Algorithm} \label{sec:adaptive}

Now we are ready to prove Theorem~\ref{thm:main-1}.

\subsection{First step: A non-adaptive algorithm}

The first step in building our adaptive solution is the construction of a {\em non-adaptive} algorithm.
\begin{theorem}\label{thm:non-adaptive}
Consider a list $R$ of $N_r$ elements and a list $B$ of $N_b$ blue elements. Then there exists a non-adaptive algorithm for red-blue dominance counting problem which requires  
$O\left(\frac{N}{B}\log_{M/B}\left(\frac{\min\{N_r,N_b\}}{B}\right)\right)$ I/Os, where $N=N_r + N_b$.
\end{theorem}
\begin{proof}
We will only give a high-level description of this algorithm. 
Most of the details are fairly standard. Without loss of 
generality, assume that $N_r=\min\{N_r,N_b\}$. 
As in distribution sort, in $O(N_r/B)$ I/Os the  list $R$ is split into $\Theta\left(\sqrt{\frac{M}{B}}\right)$ lists
$R_1, R_2,\ldots,R_f$ of roughly equal size, such that 
for any $i< j$, any element in $R_i$ is smaller than any element in $R_j$.
The order of the elements in any $R_i$ is systematic with their order in $R$. 
An element is $B$ is defined to {\em belong} to a set $R_i$ if the value of the blue 
element lies between the value of the smallest and the largest element in $R_i$. 
By performing a synchronized scan of all the $R_i$'s, in $O(N_b/B)$ I/Os, for each element in $B$ 
(say it belongs to $R_i$) we can compute the number of red points in 
$\bigcup_{i+1}^f R_i$ it dominates. 
Finally, $\forall i\in [1,f]$, we recurse on $R_i$ and the set of blue points which 
belong to $R_i$. The number of levels of recursion will be $O(\log_{M/B}\frac{N_r}{B})$.
\end{proof}

\subsection{Second step: $K$-capped structure}

Now we will solve the {\em $K$-capped  red-blue dominance counting} problem: 
Given a set $R$ of $N$ red points, a set $B$ of $N$ blue points, and a value $K$, we need to 
compute $K^{*}$, but if $K^{*} > K$, then we are allowed to report {\em failure}. 
We will prove the following result.

\begin{theorem}\label{thm:k-capped}
$K$-capped red-blue dominance counting problem can be solved using $O\left(\frac{N}{B} + \frac{N}{B}\log_{M/B}(\frac{K}{NB})\right)$ I/Os. 
\end{theorem}
Now we prove Theorem~\ref{thm:k-capped}.

\vspace{0.1 in}
\noindent
{\em Algorithm.}
Using Lemma~\ref{lem:cells}, construct red-blue cells on $R$ and $B$ with parameter $K$. 
If Lemma~\ref{lem:cells} reports a failure, then we stop the algorithm. 
Otherwise, we obtain a set of cells $C_1,\ldots, C_{\ell}$. 
For each $i\in [1,\ell]$, based on $R_i$ and $B_i$ associated with $C_i$, 
we run the non-adaptive algorithm of Theorem~\ref{thm:non-adaptive}. 
Finally, add up the count obtained from all the cells.

\vspace{0.1 in}
\noindent
{\em Analysis.}
The number of $I/Os$ performed will be bounded by 
\begin{align*}
&\sum_{i=1}^{\ell} O\left( \left(\frac{|R_i|+ |B_i|}{B} \right)\log_{M/B}\left( \frac{\min\{N_r,N_b\}}{B}\right)  \right) \\
&\leq \left(\log_{M/B}\frac{K}{NB}\right) \sum_{i=1}^{\ell} O\left( \frac{|R_i|+ |B_i|}{B} \right) \quad 
\text{by property (A)}\\
&\leq O\left(\frac{N}{B}\log_{M/B}\left(\frac{K}{NB}\right)\right) \quad \text{by property (C)}
\end{align*}

\subsection{Third step}
Using a  trick from the computational geometry literature, the solution to the 
$K$-capped red-blue dominance counting problem (Theorem~\ref{thm:k-capped}) can be used 
to efficiently solve the red-blue dominance counting problem (Theorem~\ref{thm:main-1}).

 We use Chan's guessing trick from \cite{c96b}. The algorithm is executed as a series of rounds. In round~$i$ (starting from $i=1$), we construct the $K_i$-capped structure of Theorem~\ref{thm:k-capped} 
for 
\[K_i=(NB)\cdot \left(\frac{M}{B}\right)^{\cdot2^i}\] 

If the algorithm returns the value of $K^{*}$, then we are done and the algorithm terminates. Otherwise, we 
proceed to round~$i+1$. Let $j$ be the number of rounds performed before termination. If $j=1$ then the 
number of I/Os performed is $O(N/B)$. Otherwise, 
if $j >1$ then in round $j-1$ since we reported failure, 
$K^{*} > (NB)\cdot \left(\frac{M}{B}\right)^{\cdot2^{j-1}} 
\implies 2^j < 2\log_{M/B}\frac{K^*}{NB}$.
 The total number of I/Os performed in  all the $j$ rounds is bounded by 
 $\sum_{i=1}^{j}O\left(\frac{N}{B}\log_{M/B}(\frac{K_i}{NB})\right)
 =\sum_{i=1}^{j}O\left(\frac{N}{B}\cdot 2^{i}\right)
 =O\left(\frac{N}{B}\cdot 2^{j}\right) \\
 = O\left(\frac{N}{B}\log_{M/B}(\frac{K^{*}}{NB})\right)$.

\vspace{0.1 in}
\noindent
{\em Remark.} This algorithm can be modified to match the 
 state of the art adaptive algorithms for the comparison based model~\cite{m79} and the RAM model~\cite{e15}.
This involves replacing the non-adaptive I/O-model algorithm 
of Theorem~\ref{thm:non-adaptive} with the non-adaptive 
algorithm in the comparison based model which takes $O(N\log N)$ time and 
the non-adaptive algorithm in the RAM model~\cite{cp10}. 
 
\section{The Approximation Algorithm}\label{sec:approx}

In this section we will prove Theorem~\ref{thm:approx}.
Our solution is based on an interesting combination of 
random sampling and red-blue cells.
The number of domination pairs, $K^{*}$, can lie in the 
range $[0,N^2)$. We will split the solution into three 
different cases and handle each of them separately.

\subsection{When $K^{*}\in [0,N]$}
By setting $M$ and $B$ to be appropriate constants, the I/O-model solution of 
Theorem~\ref{thm:k-capped} maps to the RAM model. We obtain the following result.
\begin{lemma}\label{lem:ram}
$K$-capped red-blue dominance counting problem can be solved in $O\left(N+N\log_2(\frac{K}{N})\right)$ time in 
the RAM model. 
\end{lemma}
Using Lemma~\ref{lem:ram} with $K=N$, we can either obtain the exact number of inversions in $O(N)$ time, or 
it will report a failure which implies that $K^* > N$.

\subsection{When $K^* \in [N, N\sqrt{N}\log N]$}

\noindent
{\em Algorithm.} The following steps are performed:

\vspace{0.1 in}
\noindent
(1) Construct the red-blue cells for parameter $K = N\sqrt{N}\log N$ using Lemma~\ref{lem:cells}.
If a failure is reported, then we conclude that $K^{*} > N\sqrt{N}\log N$ and stop 
the algorithm. Otherwise, go to the next step.

\vspace{0.1 in}
\noindent
(2) Pick $N$ samples. Each sample is a pair $(r,b)$ 
such that if $b\in B_i$ then $r\in R_i$. Each sample is picked by the following 
three stage process:
\begin{enumerate}
\item Pick a set $B_i$. A set $B_i$ is sampled with probability $\frac{|R_i||B_i|}{\sum_{i=1}^{\ell} |R_i||B_i|}$.
\item Sample a point in $B_i$. Each point in $B_i$ is sampled with probability $\frac{1}{|B_i|}$.
\item Sample a point in $R_i$. Each point in $R_i$ is sampled with probability $\frac{1}{|R_i|}$.
\end{enumerate}

\noindent
(3) Let $X$ be the number of samples which are domination pairs. 
Then we report $X\cdot C\sqrt{N}\log N$ as the answer, 
where the constant $C$ is defined later.

\begin{lemma}
Consider a pair $(r,b)$ such that $r\in R_i$ and $b\in B_i$.
The probability of the pair $(r,b)$ being picked is $\frac{1}{\sum_{i=1}^{\ell} |R_i||B_i|}$, i.e., 
each pair is picked with equal probability.
\end{lemma}
\begin{proof}
The probability of the pair $(r,b)$ being picked is 
$\frac{|R_i||B_i|}{\sum_{i=1}^{\ell} |R_i||B_i|} \times \frac{1}{|B_i|} \times \frac{1}{|R_i|}=\frac{1}{\sum_{i=1}^{\ell} |R_i||B_i|}$
\end{proof}

\begin{lemma}
The sample space is $O(N\sqrt{N}\log N)$.
In other words, $\sum_{i=1}^{\ell} |R_i||B_i| =O(N\sqrt{N}\log N)$.
\end{lemma}

\begin{proof}
We split the summation $\sum_{i=1}^{\ell} |R_i||B_i|$ into two disjoint summations:
one in which $|R_i|=\min\{ |R_i|,|B_i|\}$, and other one in which 
$|B_i|=\min\{ |R_i|,|B_i|\}$. Consider the first summation:
\begin{align*}
\sum_{i=1}^{\ell} |R_i||B_i| &\leq O(K/N)\sum_{i=1}^{\ell} |B_i|\quad \text{by property (A)} \\
&\leq O(K/N) \cdot O(N) \quad \text{by property (C)}\\
  &= O(N\sqrt{N}\log N)   
\end{align*}
The same bound can be shown for the other summation as well.
\end{proof}

\begin{lemma}
For a fixed constant $c$, with high probability $1-1/N^c$, the estimate 
will lie in the range \\ $\left[ \left(1-\frac{\log N}{N^{1/4}}\right)K^{*},\left(1+ \frac{\log N}{N^{1/4}}\right)K^{*}\right]$.
\end{lemma}
\begin{proof}
Recall that $X$ is the number of domination pairs picked in the $N$ samples. 
For $i \in [1,N]$, define $X_i=1$ if the $i$-th sample picked is a 
domination pair; otherwise $X_i=0$. Therefore, $X=\sum_{i=1}^N X_i$. 
The expected value of $X$, i.e., $E[X]$ will be equal to \\ $N\cdot\frac{K^{*}}{\sum_{i=1}^{\ell} |R_i||B_i|}=N\cdot\frac{K^{*}}{CN\sqrt{N}\log N}=\frac{K^{*}}{C\sqrt{N}\log N}$, where $C$  is the constant inside $O(N\sqrt{N}\log N)$.

To apply Chernoff bounds, we need to perform the following set of calculations. 
Set a parameter $\varepsilon = \frac{\log N}{N^{1/4}}$ and use the fact that $K^{*} \geq N$, to observe that 
\begin{align*}
\varepsilon^2 E[X] = \varepsilon^2 \frac{K^{*}}{C\sqrt{N}\log N} > \varepsilon^2 \frac{\sqrt{N}}{C\log N} =
\frac{\log N}{C} \\
\end{align*}
By applying Chernoff bounds, we get 
\begin{align*}
\textbf{Pr}\bigg[ \bigg|X-E[X]\bigg|>\varepsilon E[X]\bigg] &< e^{-\Omega(\varepsilon^2 E[X])} < e^{-\Omega(\log N)} < N^{-c}
\end{align*}
\end{proof}

\subsection{When $K^{*} \in [N\sqrt{N}\log N, N^2]$}

\noindent
{\em Algorithm.} The following steps are performed: 

\vspace{0.1 in}
\noindent
(1) Pick $N$ random samples. Each sample is of the form $(r,b)$ where $r\in R$ and $b\in B$. 
Each red point in $R$ is picked with probability $\frac{1}{N}$ and each blue point in $B$ 
is picked with probability $\frac{1}{N}$. 

\vspace{0.1 in}
\noindent
(2) Let $X$ be the number of samples which are domination pairs. 
Then we report $X\cdot N$ as the answer.

\begin{lemma}
Let $c$ be a sufficiently large constant. Then with high probability $1-1/N^c$, the estimate 
will lie in the range $\left[ \left(1-\frac{1}{N^{1/4}}\right)K^{*},\left(1+ \frac{1}{N^{1/4}}\right)K^{*}\right]$.
\end{lemma}
\begin{proof}
Let $X$ be the number of domination pairs picked in the $N$ samples. 
For $i \in [1,N]$, define $X_i=1$ if the $i$-th sample picked is a 
domination pair; otherwise $X_i=0$. Therefore, $X=\sum_{i=1}^N X_i$. 
Now, $E[X]=N\cdot\frac{K^{*}}{N^2}=\frac{K^{*}}{N}$.

To apply Chernoff bounds, we need to perform the following set of calculations. 
Set a parameter $\varepsilon = 1/N^{1/4}$ and use the fact that $K^{*} \geq N\sqrt{N}\log N$, to observe that 
\begin{align*}
\varepsilon^2 E[X] = \varepsilon^2 \frac{K^{*}}{N} > \varepsilon^2 \frac{N\sqrt{N}\log N}{N} =\log N \\
\end{align*}
By applying Chernoff bounds, we get 
\begin{align*}
Pr\bigg[ \bigg|X-E[X]\bigg|>\varepsilon E[X]\bigg] &< e^{-\Omega(\varepsilon^2 E[X])} < e^{-\Omega(\log N)} < N^{-c}
\end{align*}
\end{proof}

\bibliographystyle{plain}
\bibliography{./ref}

\section*{Appendix: I/O-model}
In this model~\cite{av88}, a machine is equipped with $M$ words of main memory, and a disk that has been formatted into {\em blocks} of $B$ words each. The values of $M$ and $B$ satisfy $M \ge 2B$. An I/O either reads a disk block into memory, or writes $B$ words of memory into a disk block. The {\em time} of an algorithm is measured in the number of I/Os performed, while the {\em space} is measured in the number of disk blocks occupied.

\end{document}

%% file: yf-def.tex
\newtheorem{theorem}{Theorem}
\newtheorem{lemma}{Lemma}

\def\-{\mbox{-}}

\def\*{\star}

